\documentclass[letterpaper, 10 pt, conference]{ieeeconf}  
\pdfoutput=1

\IEEEoverridecommandlockouts                              

\usepackage{cite}
\usepackage{mathrsfs} 
\usepackage{ifpdf}

\usepackage{graphicx}
\usepackage{amsmath}
\newcommand\norm[1]{\left\lVert#1\right\rVert}
\usepackage{amsmath , amsfonts , amssymb}
\newcommand\myeqa{\mathrel{\stackrel{\makebox[0pt]{\mbox{\normalfont\tiny (a)}}}{=}}}
\newcommand\myeqb{\mathrel{\stackrel{\makebox[0pt]{\mbox{\normalfont\tiny (b)}}}{\geq}}}
\newcommand\myeqc{\mathrel{\stackrel{\makebox[0pt]{\mbox{\normalfont\tiny (c)}}}{=}}}
\newcommand\myeqd{\mathrel{\stackrel{\makebox[0pt]{\mbox{\normalfont\tiny (d)}}}{=}}}
\newcommand\myeqe{\mathrel{\stackrel{\makebox[0pt]{\mbox{\normalfont\tiny (e)}}}{=}}}
\newcommand\myeqf{\mathrel{\stackrel{\makebox[0pt]{\mbox{\normalfont\tiny (f)}}}{\geq}}}
\newcommand\myeqg{\mathrel{\stackrel{\makebox[0pt]{\mbox{\normalfont\tiny (g)}}}{=}}}
\newcommand\myeqh{\mathrel{\stackrel{\makebox[0pt]{\mbox{\normalfont\tiny (h)}}}{=}}}
\newcommand\myeqi{\mathrel{\stackrel{\makebox[0pt]{\mbox{\normalfont\tiny (i)}}}{=}}}
\newcommand\myeqj{\mathrel{\stackrel{\makebox[0pt]{\mbox{\normalfont\tiny (j)}}}{=}}}
\newcommand\myeqk{\mathrel{\stackrel{\makebox[0pt]{\mbox{\normalfont\tiny (k)}}}{=}}}
\newcommand\myeql{\mathrel{\stackrel{\makebox[0pt]{\mbox{\normalfont\tiny (l)}}}{=}}}
\newcommand\myeqm{\mathrel{\stackrel{\makebox[0pt]{\mbox{\normalfont\tiny (m)}}}{=}}}
\newcommand\myeqn{\mathrel{\stackrel{\makebox[0pt]{\mbox{\normalfont\tiny (n)}}}{=}}}
\newcommand\myeqo{\mathrel{\stackrel{\makebox[0pt]{\mbox{\normalfont\tiny (o)}}}{=}}}
\newcommand\myeqp{\mathrel{\stackrel{\makebox[0pt]{\mbox{\normalfont\tiny (p)}}}{\geq}}}
\newcommand\myeqq{\mathrel{\stackrel{\makebox[0pt]{\mbox{\normalfont\tiny (q)}}}{\geq}}}
\newcommand\myeqr{\mathrel{\stackrel{\makebox[0pt]{\mbox{\normalfont\tiny (r)}}}{\geq}}}
\newcommand\myeqs{\mathrel{\stackrel{\makebox[0pt]{\mbox{\normalfont\tiny (s)}}}{=}}}
\newcommand\myeqt{\mathrel{\stackrel{\makebox[0pt]{\mbox{\normalfont\tiny (t)}}}{\geq}}}
\newcommand\myequ{\mathrel{\stackrel{\makebox[0pt]{\mbox{\normalfont\tiny (u)}}}{=}}}
\newcommand\myeqv{\mathrel{\stackrel{\makebox[0pt]{\mbox{\normalfont\tiny (v)}}}{\geq}}}
\newcommand\myeqw{\mathrel{\stackrel{\makebox[0pt]{\mbox{\normalfont\tiny (w)}}}{=}}}
\newcommand\myeqx{\mathrel{\stackrel{\makebox[0pt]{\mbox{\normalfont\tiny (x)}}}{\geq}}}
\newcommand\myeqy{\mathrel{\stackrel{\makebox[0pt]{\mbox{\normalfont\tiny (y)}}}{=}}}
\newcommand\myeqz{\mathrel{\stackrel{\makebox[0pt]{\mbox{\normalfont\tiny (z)}}}{=}}}
\newcommand\myeqaa{\mathrel{\stackrel{\makebox[0pt]{\mbox{\normalfont\tiny (aa)}}}{\geq}}}
\newcommand\myeqab{\mathrel{\stackrel{\makebox[0pt]{\mbox{\normalfont\tiny (ab)}}}{=}}}

\title{\LARGE \bf
Interplay Between Transmission Delay, Average Data Rate, and
Performance in Output Feedback Control over Digital Communication
Channels*
}

\author{Mohsen Barforooshan$^{1}$, Jan \O stergaard$^{1}$, and Milan S. Derpich$^{2}$
\thanks{*This work has received funding from VILLUM FONDEN Young
	Investigator Programme, under grant agreement No. 19005.}
\thanks{$^{1}$Mohsen Barforooshan and Jan \O stergaard are with the Department of
	Electronic Systems,
        Aalborg University, Nieles Jernes Vej 12, DK-9220,
        Aalborg, Denmark
        {\tt\small \{mob,jo\}@es.aau.dk}}%
\thanks{$^{2}$Milan S. Derpich is with the Department of Electronic Engineering, Universidad T\'ecnica Federico Santa Mar\'ia,
       Casilla 110-V, Valpara\'iso, Chile
        {\tt\small milan.derpich@usm.cl}}%
}

\begin{document}
\maketitle
\thispagestyle{empty}
\pagestyle{empty}

\begin{abstract}

The performance of a noisy linear time-invariant (LTI) plant, controlled over a noiseless digital channel with transmission delay, is investigated in this paper. The rate-limited channel connects the single measurement output of the plant to its single control input through a causal, but otherwise arbitrary, coder-controller pair. An infomation-theoretic approach is utilized to analyze the minimal average data rate required to attain the quadratic performance when the channel imposes a known constant delay on the transmitted data. This infimum average data rate is shown to be lower bounded by minimizing the directed information rate across a set of LTI filters and an additive white Gaussian noise (AWGN) channel. It is demonstrated that the presence of time delay in the channel increases the data rate needed to achieve a certain level of performance. The applicability of the results is verified through a numerical example. In particular, we show by simulations that when the optimal filters are used but the AWGN channel (used in the lower bound) is replaced by a simple scalar uniform quantizer, the resulting operational data rates are at most around 0.3 bits above the lower bounds.        

\end{abstract}

\section{INTRODUCTION}\label{sec1}

Taking communication imperfections into account for
analysis and design has proved to be an overwhelming topic
within the area of control theory during recent years. Among those imperfections,
time delay, packet dropout and bit rate constraint (quantization) are
prominent ones, which may worsen the performance and even bring destabilization to networked control systems (NCSs) \cite{zhang2013network},\cite{baillieul2007control}.

Rate-constrained NCSs are generally studied form two points of view; control theory and information theory. Regarding the first viewpoint, classical nonlinear control methods are deployed (see, e.g., \cite{delchamps1990stabilizing,wong1999systems} as early results). As a recent contribution,  \cite{almakhles2015stability} obtains stability conditions in terms of the quantizer's step size by using sliding mode analysis. For the second point of view, the key idea  is extending information-theoretic notions to the case of closed-loop control. Regarding stabilization, early results can be found in \cite{nair2004stabilizability}. Recently, in \cite{lupu2015information}, mean square stability (MSS) conditions for an unstable human-in-the-loop system are stated in terms of a lower bound on the information rate between control input and system output.

For performance, efforts fall into either control-based or information theory-based
approach as well. Extending Bode integral to discrete linear time-periodic multirate systems is carried out in \cite{zhao2015bode}. A packetized predictive control strategy is studied in \cite{ostergaard2016multiple} where the Markov jump linear systems (MJLSs) theory is used for deriving an upper bound on the bit rate required to attain a desired performance level, under the circumstances of entropy-coded dithered quantization (ECDQ\footnote{With some abuse of notation, and depending on the context, ECDQ is also to be understood as entropy-coded dithered quantizer}) with multiple descriptions.

Studies analyzing system performance from an information-theoretic viewpoint are less abundant in the literature. Primary results are presented in \cite{martins2008feedback}. In this work, for a discrete-time linear time-invariant (LTI) plant, the well-known Bode's integral is extended to the case with causal rate-limited arbitrary feecback. Along the lines of \cite{martins2008feedback},  research reported in \cite{silva2011framework, silva2013characterization} has investigated bounds on the minimum data rate which is needed to attain a quadratic performance level in NCSs with delay-free channel. While the design approach proposed in \cite{silva2011framework} is a second-stage one,  coding and control are designed jointly in \cite{silva2013characterization}. For the lower bound, \cite{silva2013characterization} shows that the rate-constrained optimization to find desired infimal data rate over causal but otherwise arbitrary coder-controller pairs, is reduced to a convex SNR-constrained optimization over an auxiliary LTI feedback loop with additive white Gaussian noise (AWGN). Based on the SNR-constrained approach, \cite{silva2013characterization} proposes an ECDQ-based linear coding scheme which leads to an upper bound at most 1.25 bits away from the obtained lower bound. Then the main idea followed in \cite{silva2011framework,silva2013characterization} (that is, minimizing directed information rate to get lower bound and entropy coding for upper bound) is applied to LQG control of a fully-observable multiple-input multiple-otput (MIMO) plant in \cite{tanaka2016rate}. The  authors of \cite{tanaka2016rate} deploy semidefinite programming (SDP) to solve corresponding optimization problems.

In this paper, we address output feedback control of an NCS comprised of a noisy LTI plant and a causal encoder-controller-decoder set connected through a noiseless digital channel with a constant transmission delay. More specifically, the problem  is obtaining the bounds on minimal average data rate required to guarantee that the steady-state variance of an error signal does not become larger than a certain value. Motivated by its merits such as simplicity and practical appeal, we use the approach pursued in \cite{silva2011framework,silva2013characterization} to gain outer bounds and build upon \cite{silva2013characterization}. However, the main departure of this work from  \cite{silva2013characterization} is considering a channel which is not delay-free. So, as the first contribution, we rederive fundamental information inequalities of the system under the delay assumption. Secondly, we characterize the trade-off among performance, delay, and minimal desired average data rate. It is shown through a numerical example that greater transmission delay necessitates greater minimal average data rate needed to guarantee achieveing the considered quadratic level of performance. Simulation indicates that by deploying a simple scalar uniform quantizer in the LTI architecture that gives the lower bound, the quadratic performance is attained by operational average data rates at most $0.3$ bits away from the lower bound.  

The outline of this work is as follows. Section
II presents the notation and some preliminaries. Then the problem of interest is formalized in Section III. Section IV is dedicated to the lower bound characterization. An illustrative numerical simulation is provided in section V. Finally, Section VI concludes the paper.                       
\section{NOTATION AND PRELIMINARIES}\label{sec2}
The set of real numbers is denoted by $\mathbb{R}$ with subset ${\mathbb{R}}^{+}$ as the set of strictly positive real numbers. $\mathbb{N}$ represents the set of natural numbers, based upon which ${\mathbb{N}_{0}}=\mathbb{N}\cup\{0\}$ is defined. Furthermore, $k$ is the time index and for random processes considered in this paper, $k\in{\mathbb{N}_{0}}$ holds.
 Magnitude and $H_2$-norm of a signal are symbolized by $|.|$  and ${\norm{.}}_{2}$, respectively. Furthermore, the set ${{\mathcal{U}}_\infty}$ is defined as the
set of all proper and real rational stable transfer functions with inverses that are
stable and proper as well. $\mathcal{E}$ denotes the expectation operator and $\log$ stands for the natural logarithm. The entry of matrix $S$ on the $i$-th row and $j$-th column is denoted by $[S]_{i,j}$. Moreover, ${\lambda}_{min}(S)$ and ${\lambda}_{max}(S)$ represent eigenvalues of $S$ with the smallest and largest magnitude, respectively.

All random variables and processes in this paper are assumed to be vector valued, unless otherwise stated. The mutual information between random variables $V$ and $W$ is defined as $I(V;W)\triangleq{H(V)-H(V\mid{W})}$ whit $H(.)$ denoting the entropy. If those variables are conditioned upon another rendom variable, say $Y$, then  $I(V;W\mid{Y})\triangleq{H(V\mid{Y})-H(V\mid{W,Y})}$ defines the conditional mutual information between $V$ and $W$ given $Y$ \cite{cover2012elements}. A random process $\xi$ is said to be asymptotically wide-sense stationary (AWSS) if it satisfies $\lim_{k\to\infty}{\mathcal{E}[\xi(k)]}={\nu}_{\xi}$ and $\lim_{k\to\infty}{\mathcal{E}[(\xi(k+\tau)-\mathcal{E}[\xi(k+\tau)]){(\xi(k)-\mathcal{E}[\xi(k)])}^{T}]}={R}_{\xi}(\tau)$ hold, where ${\nu}_{\xi}$ is a constant. ${C_\xi}={R}_{\xi}(0)$ denotes the corresponding steady-state covariance matrix upon which the steady-state variance of $\xi$ is defined as $\sigma_\xi^2\triangleq\mathrm{trace}({C_\xi})$. The covariance matrix for a scalar random sequence ${x}_{1}^{k}\!\!\!\triangleq\!\!{[x(1)\dots x(k)]^T}$ is defined as $C_{{x}_{1}^{k}}={\mathcal{E}[({{x}_{1}^{k}}-\mathcal{E}[{{x}_{1}^{k}}]){({{x}_{1}^{k}}-\mathcal{E}[{{x}_{1}^{k}}])}^{T}]}$. 
 \begin{figure}[thpb]
 	\centering
 	\includegraphics[width=6cm,height=3.5cm]{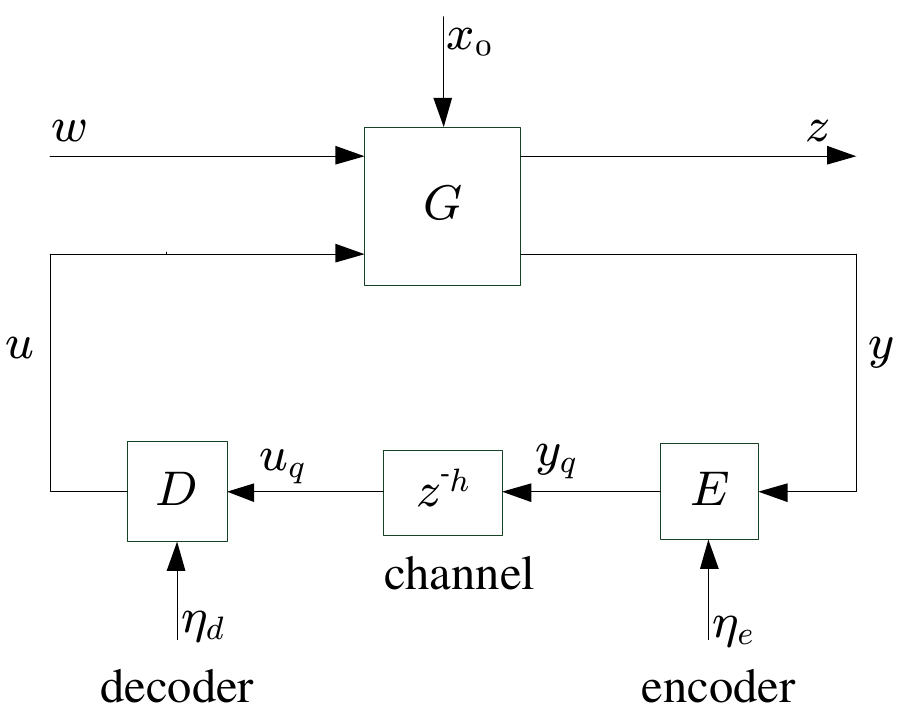}
 	\caption{Considered NCS}
 	\label{fig1}
 \end{figure}
Considering ${P_n},{Q_n}\in{{\mathbb{R}}^{n\times n}}$ as two square matrices, the sequences $\{P_n\}_{n=1}^{\infty}$ and $\{Q_n\}_{n=1}^{\infty}$ are asymptotically equivalent if and only if the following holds for finite $\varrho$:  
\begin{IEEEeqnarray*}{rl}
\begin{split}
	&\lim_{n\to\infty}\frac{1}{n}\sum_{i=1}^{n}\sum_{j=1}^{n}{{|{[{P_n}-{Q_n}]}_{i,j}|}^{2}}=0\\
&|{\lambda}_{max}(P_n)|,|{\lambda}_{max}(Q_n)|\leq{\varrho},\quad\forall{n\in{\mathbb{N}}} 
	\end{split}
\end{IEEEeqnarray*}
\section{PROBLEM STATEMENT AND SETUP}
 The structure considered in this work can be found in Fig.~{\ref{fig1}} where $G$ is an LTI plant with $u\in\mathbb{R}$ as control input and $y\in\mathbb{R}$ as sensor output. Moreover, there is a disturbance represented by $w\in \mathbb{R}^{n_{w}}$ and the output signal $z\in \mathbb{R}^{n_{z}}$ upon which the  desired performance is characterized. The plant has the following transfer-function matrix description:
\begin {equation}\label{eq1}
\left[
\begin{array}{c}
	z\\
	y
\end{array}\right]=\left[
\begin{array}{lr}
	G_{11}&G_{12}\\
	G_{21}&G_{22}
\end{array}\right]\left[
\begin{array}{c}
	w\\
	u
\end{array}\right],
\end{equation}
in which every $G_{ij}$ is proper and of suitable dimension. The input alphabet of the channel is represented by $\mathcal{A}$ and is defined as a countable set of prefix-free binary words. Due to the delay, the output of the channel $u_q(k)$ follows ${u_q}(k)={y_q}(k-h)$ for $k\geq{h}$ where $y_q(k)$ belongs to $\mathcal{A}$.
 The average data rate across the channel is specified as follows:
\begin{equation}\label{eq2}
\mathcal{R}
\triangleq\lim_{k \to \infty}
\frac{1}{k}
\sum_{i=0}^{k-1}R(i),
\end{equation}
where $R(i)$ denotes the expected length of the $i$-th binary word ${y_q}(i)$.  
The channel input is provided by the encoder $E$ based on the following dynamics:
\begin{equation}\label{eq3}
{y_q(k)}={E_k}({y^k},{\eta_{e}^{k}}),
\end{equation}  
in which ${{\eta_{e}}(k)}$ is the side information at time $k$ at the encoder with ${E_k}$ representing an arbitrary (possibly nonlinear or time-varying) deterministic mapping. It should be noted that $\beta^k$ is a shorthand for $[\beta(0),\cdots,\beta(k)]$. On the decoder side, we have
\begin{equation}\label{eq4}
{u(k)}=
\begin{cases}
{D_k}({\eta_{d}^{k}}), &0\leq{k}<h,\\
{D_k}({{y}_{q}^{k-h}},{\eta_{d}^{k}}), &k\geq{h}.
\end{cases}
\end{equation}
${{D}_k}$ is assumed to be an arbitrary deterministic mapping, like ${{E}_k}$, and ${\eta_{d}(k)}$ signifies the side information available at the decoder at time $k$. It should be emphasized that $E$ and $D$ in Fig.~\ref{fig1} are possibly time-varying or nonlinear causal systems. 
\newtheorem{Assumption}{Assumption}[section]
\begin{Assumption}\label{ass1}
	The plant $G$ is LTI, proper and free of unstable hidden modes. Moreover, the open-loop transfer function from $u$ to $y$ is single-input single-output (SISO) and strictly proper. The disturbance signal, $w$, is a zero-mean white noise with identity covariance matrix ${C_w}=I$ and jointly Gaussian with ${x_0}=[x(-h),\cdots,x(0)]^T$, the initial condition, having finite differential entropy.
\end{Assumption}
\begin{Assumption}\label{ass2}
Each of processes $\eta_e$ and $\eta_d$ is jointly independent of $({x_0},w)$. So regarding the dynamics of the system,  $I(u(k);{y^{k-h}}\mid{u^{k-1}})=0$  holds for $0\leq{k}<h$. Moreover, upon knowledge of $u^{i}$ and $\eta_{d}^{i}$, the decoder is invertible. It means that there exists a deterministic mapping $Q_{i}$ such that $u_{q}^{i}=Q_{i}(u^i,\eta_{d}^{i})$.     
\end{Assumption}
\newtheorem{Remark}{Remark}[section]
\begin{Remark}\label{rem1}
Regarding the aforementioned setup, invertibility is a reasonable assumption for the decoder. It can be proved that for the architecture of Fig.~\ref{fig1}, any encoder and non-invertible decoder pair, with mappings $E_k$ and $D_k$, can be replaced by another pair with the same input-output relationship and lower average data rate where the decoder is invertible. Due to space limitations, we eliminate the proof.
\end{Remark}

\newtheorem{Definition}{Definition}[section]
\begin{Definition}\label{def1}
A scalar AWSS process $x$ converging to a wide sense stationary process $\bar{x}$ is called strongly asymptotically wide-sense stationary (SAWSS) if asymptotic equivalence holds between their covariance matrices, ${\{C}_{x_1^n}\}_{n=1}^{\infty}$ and ${\{C}_{{\bar{x}}_1^n}\}_{n=1}^{\infty}$. Accordingly, in an SAWSS NCS, all covariance and cross-covariance matrices of internal signals are asymptotically equivalent to their counterparts in processes to which they converge.
\end{Definition}

Now, suppose that Assumption~\ref{ass1} holds. Let $D_{inf}(h)$ denote the infimum steady-state variance of the output $z$ over all settings $u(k)={\mathcal{K}}_{k}(\gamma^{k})$ for $0\leq{k}<h$ and $u(k)={\mathcal{K}}_{k}(y^{k-h})$ for $k\geq{h}$ with $\gamma^{k}$ independent of $x_0$ and $w$. Then the problem of our interest is finding 
\begin{equation}\label{eq5}
	{\mathcal{R}}(D)=\inf_{{{\sigma}_{z}^{2}}\leq{D}}
	\mathcal{R}
\end{equation}
for any $D\in{(D_{inf}(h),\infty)}$, where the search is to be restricted to encoders with mapping $E_k$ and decoders with mapping $D_k$ which satisfy Assumption~\ref{ass2} and make the NCS of Fig.~\ref{fig1} SAWSS. Moreover, ${{\sigma}_{z}^{2}}$ denotes the steady-state variance of $z$. It can be proved that the optimization problem in (\ref{eq5}) is feasible if ${D}\in({D_{inf}}(h),\infty)$ (see Appendix \ref{suba1}).
 \section{MAIN RESULTS} 
This section shows that in order to obtain a lower bound on ${\mathcal{R}}(D)$, one can minimize the directed information over an auxiliary coding scheme formed of LTI filters and an AWGN channel with feedback and delay. Regarding this optimization problem, inequalities and identities in \cite{silva2013characterization} will be extended to the case with a channel subject to delay in the following.
\newtheorem{Theorem}{Theorem}[section]
\begin{Theorem}\label{th1}
For the feedback loop  depicted in Fig.~\ref{fig1} and satisfying Assumptions~\ref{ass1} and \ref{ass2}, the following holds: 
\begin{equation}\label{eq6}
\mathcal{R}
\geq{{I_{\infty}^{(h)}}(y\to{u})}=\lim_{k \to \infty}
\tfrac{1}{k}
\Sigma_{i=0}^{k-1}I(u(i);{y^{i-h}}\mid{u^{i-1}}),
\end{equation}
in which $I(.;.\mid{.})$ indicates conditional mutual information. Moreover, as defined in \cite{derpich2013fundamental}, ${{I_{\infty}^{(h)}}(y\to{u})}$ denotes the directed information  rate across the forward channel from $y$ to $u$ with delay $h$. For the proof, see Appendix ~\ref{suba2}.	
\end{Theorem}

 Now, a lower bound can be derived on the directed information across the coding scheme of Fig.~\ref{fig1}.

\newtheorem{Lemma}{Lemma}[section]
\begin{Lemma}\label{lemma41}
	For the NCS of Fig.~\ref{fig1}, assume that $(x(0),w,u,y)$ form a jointly second-order set of processes and that Assumptions~\ref{ass1} and \ref{ass2} hold. Moreover, take $y_G$ and $u_G$ into account as the Gaussian  counterparts of $y$ and $u$ where $(x(0),w,u_G,y_G)$ are jointly Gaussian with the same first-and second-order (cross-) moments as $(x(0),w,u,y)$. Then ${{I_{\infty}^{(h)}}(y\to{u})\geq{I_{\infty}^{(h)}}({y_G}\to{u_G})}$.    
\end{Lemma}
\begin{proof}
	The following inequalities and identities will justify the claim: 
	\begin{IEEEeqnarray}{rl}\label{eq7}
		\begin{split}
		&\Sigma_{i=0}^{k-1}{I({{u}(i);y^{i-h}}\mid{{u^{i-1}}})}\myeqa{I(x(0),{{w}^{k-1}};{u}^{k-1})}\\
		&\myeqb{I(x(0),{{w}^{k-1}};{u}_{G}^{k-1})}\\
		&\myeqc\Sigma_{i=0}^{k-1}{I(x(0),{{w}^{k-1}};{{u}_{G}}(i)\mid{u_{G}^{i-1}})}\\
		&\myeqd\Sigma_{i=0}^{k-1}{I(x(0),{{w}^{i}};{{u}_{G}}(i)\mid{u_{G}^{i-1}})}\\
		&\myeqe\Sigma_{i=0}^{k-1}{[I(x(0),{{w}^{i}},{{y}_{G}^{i-h}};{{u}_{G}}(i)\mid{u_{G}^{i-1}})}\\
		&\negmedspace {}\qquad\qquad\qquad-{I({{y}_{G}^{i-h}};{{u}_{G}}(i)\mid{u_{G}^{i-1}},x(0),{{w}^{i}})]}\\
		&\myeqf\Sigma_{i=0}^{k-1}{[I(x(0),{{w}^{i-1}},{{y}_{G}^{i-h}};{{u}_{G}}(i)\mid{u_{G}^{i-1}})}\\
	    &\negmedspace {}\qquad\qquad\qquad-{I({{y}_{G}^{i-h}};{{u}_{G}}(i)\mid{u_{G}^{i-1}},x(0),{{w}^{i}})]}\\
		&\myeqg\Sigma_{i=0}^{k-1}{I(x(0),{{w}^{i-1}},{{y}_{G}^{i-h}};{{u}_{G}}(i)\mid{u_{G}^{i-1}})}\\
		&\myeqh\Sigma_{i=0}^{k-1}{I({{y}_{G}^{i-h}};{{u}_{G}}(i)\mid{u_{G}^{i-1}})}\\
		&\negmedspace {}\qquad\qquad\qquad+{I(x(0),{{w}^{i-1}};{{u}_{G}}(i)\mid{u_{G}^{i-1}},{{y}_{G}^{i-h}})}\\
		&\myeqi\Sigma_{i=0}^{k-1}{I({{y}_{G}^{i-h}};{{u}_{G}}(i)\mid{u_{G}^{i-1}})},
	\end{split}
	\end{IEEEeqnarray}
	where (a) is obtained by a slight modification in Lemma B.4 of \cite{silva2013characterization} based on the result of Theorem 1 in \cite{derpich2013fundamental}. Lemma B.1 in \cite{silva2013characterization} will give (b) because as mentioned in Assumption \ref{ass1}, $x_0$ and $w^k$ are jointly Gaussian.
	Moreover, (c) is a result of the chain rule applied on the conditional mutual information. (d) is concluded because regarding (52b) in Theorem B.3 of \cite{silva2013characterization},  which holds for the considered NCS in this paper, the Markov chain  ${{u}_{G}}(i)-{u_{G}^{i-1}},x(0),{w^i}-{{w}_{i+1}^{k-1}}$ holds . Furthermore, the chain rule of conditional mutual information gives (e), and (f) is caused by a property of mutual information. (g) is concluded according to Assumption \ref{ass1} and ${{y}_{G}^{i}}$ being a deterministic function of ${u_{G}^{i-1}},x(0)$ and ${{w}^{i}}$. The chain rule will yield (h). Finally, (i) is concluded since regarding (52a) in Theorem B.3 of \cite{silva2013characterization} and Assumption \ref{ass1}, the validity of the Markov chain ${{u}_{G}}(i)- {u_{G}^{i-1}},{{y}_{G}^{i-h}}-x(0),{{w}^{i-1}}$ is verified.   The proof is complete.
\end{proof}

What follows will relate the directed information from $y_G$ to $u_G$ to their corresponding power spectral densities:
\begin{Lemma}\label{lemma42}
Consider $y$ and $u$ as jointly Gaussian AWSS processes. Moreover, suppose that $u$ is SAWSS with $|{{{\lambda}_{min}}(C_{{u}_{1}^{n}})}|\geq{\mu}$, $\forall{n}\in{\mathbb{N}}$ where $\mu>0$. Then the following can be obtained:  
\begin{equation}\label{eq8}
	{{I_{\infty}^{(h)}}(y\to{u})}=\frac{1}{4\pi}\int_{-\pi}^{\pi} \log\big(\frac{{S_{\check{u}}}(e^{j\omega})}{{\sigma}_{\psi}^{2}}\big)d\omega,
\end{equation}	
in which $\psi$ is a Gaussian AWSS process with independent samples defined as:
	\begin{figure}[thpb]
		\centering
		\includegraphics[width=6cm,height=3.5cm]{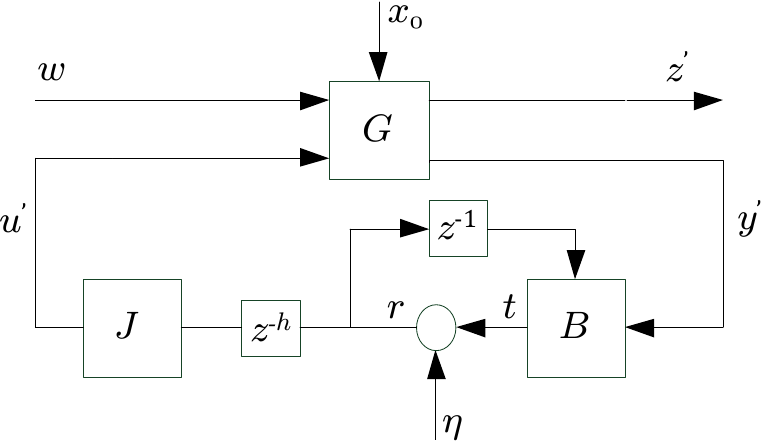}
		\caption{Auxiliary LTI NCS}
		\label{fig2}
	\end{figure}
\begin{equation}\label{eq9}
\psi(k)\triangleq{u(k)-\tilde{u}(k)},
\tilde{u}(k)\triangleq{E[u(k)\mid{{y^{k-h}},{u^{k-1}}}]}.
\end{equation} 
Moreover, ${S_{\check{u}}}$ represents the steady-state power spectral density of $u$.
 \end{Lemma}
\begin{proof}
Having Gaussianity and joint AWSS-ness of $(u,y)$ in mind and based on Theorem 2.4 in \cite{porat1994digital} with a little modification, we can conclude that $\psi$ is Gaussian and AWSS as well. We start by the following equalities:	
\begin{IEEEeqnarray}{rl}\label{eq10}
\begin{split}	
&{I({{u}(i);y^{i-h}}\mid{{u^{i-1}}})}\myeqj{h(u(i)\mid{u^{i-1}})-h(u(i)\mid{y^{i-h},u^{i-1}})}\\
&\myeqk{h(u(i)\mid{u^{i-1}})-h(\psi(i)+\tilde{u}(i)\mid{y^{i-h},u^{i-1}})}\\
&\myeql{h(u(i)\mid{u^{i-1}})-h(\psi(i)\mid{y^{i-h},u^{i-1}})}\\
&\myeqm{h(u(i)\mid{u^{i-1}})-h(\psi(i))},
\end{split}
\end{IEEEeqnarray}
where (j) comes from the definition of mutual information and (k) from the definition of $\psi$. Furthermore, (l) can be concluded based on Property 2 in \cite{silva2013characterization} and (\ref{eq9}) where $\tilde{u}(i)$ is a deterministic function of ${y^{i-h}}$ and ${u^{i-1}}$. Due to the independence between $\psi(i)$ and $({y^{i-h}},{u^{i-1}})$, and Property 3 in \cite{silva2013characterization}, (m) is obtained. So the directed information rate can therefore be rewritten as follows:
\begin{IEEEeqnarray}{rl}\label{eq11}
	\begin{split}
&{{I_{\infty}^{(h)}}(y\to{u})}=\lim_{k \to \infty}
\frac{1}{k}
\sum_{i=0}^{k-1}\{{h(u(i)\mid{u^{i-1}})-h(\psi(i))}\}\\
&\myeqn\lim_{k \to \infty}\frac{1}{k}{h(u^{k-1})}-\lim_{k \to \infty}{h(\psi(k))}\\
&\myeqo{\frac{1}{4\pi}\int_{-\pi}^{\pi} \log(2\pi{e}{S_{\check{u}}}(e^{j\omega}))d\omega}- \frac{1}{2}\log({2\pi{e}{{\sigma}_{\psi}^{2}}}),
	\end{split}
\end{IEEEeqnarray}
where (n) follows from the chain rule of the differential entropy and $\psi(k)$ being independent of ${\psi}^{k-1}$. Since the process $u$ is SAWSS with $|{{{\lambda}_{min}}(C_{{u}_{1}^{n}})}|\geq{\mu}$, $\forall{n}\in{\mathbb{N}}$ for some $\mu>0$, Lemma B.5 in \cite{silva2013characterization} will approve the validity of the leftmost term in (o). The rightmost term is self-explanatory because $\psi$ is Gaussian and AWSS.
\end{proof}

It follows from Theorem~\ref{th1}, Lemma~\ref{lemma41} and Lemma~\ref{lemma42} that the rate-performance pair yielded by any encoder-decoder scheme  which renders the NCS SAWSS is attainable with a lower rate by a coding scheme comprised of LTI filters and an AWGN noise source. Such a scheme is depicted in Fig. 2.

The NCS of Fig.~\ref{fig2} is defined under the same conditions (Assumption \ref{ass1}) as the main system in Fig.~\ref{fig1} except for one thing; the arbitrary mappings are replaced by proper LTI filters $B$ and $J$. Moreover, the communication channel is a delayed AWGN channel with noiseless one-sample-delayed feedback. The dynamics of this auxiliary coding scheme can be summarized as follows:	
\begin{equation}\label{eq12}
	u'=J{z^{-h}}r,\quad r=t+\eta,\quad t=B\mathrm{diag}\{{z^{-1}},1\} \begin{bmatrix}
		r\\
		y'
	\end{bmatrix},
\end{equation}
in which $\eta$ is the AWGN with zero mean and variance ${{\sigma}_{\eta}^{2}}$. This noise is assumed to be independent of $(x_0,w)$. Additionally, we suppose that the initial state of $B$, $J$, and the delay are deterministic.
\begin{Theorem}\label{thm2}
	For the NCS depicted in Fig.~\ref{fig1} and satisfying Assumptions~\ref{ass1} and \ref{ass2}, $\mathcal{R}(D)$ is lower-bounded as follows if ${D}\in({D_{inf}}(h),\infty)$: 
	\begin{equation}\label{eq13}
	{{\mathcal{R}}(D)}\geq{{{\vartheta}'}_u}(D)\triangleq\inf_{{{\sigma}_{z'}^{2}}\leq{D}}{\frac{1}{4\pi}\int_{-\pi}^{\pi} \log\big(\frac{{S_{u'}}(e^{j\omega})}{{\sigma}_{\eta}^{2}}\big)},
	\end{equation}
where the feasible set for the optimization problem defining ${{{\vartheta}'}_u}(D)$ is all LTI filters $B$ and the noise $\eta$ with ${{\sigma}_{{\eta}}^{2}}\in{{\mathbb{R}}_{+}}$ rendering the feedback loop of Fig.~\ref{fig2} internally stable and well-posed when $J=1$. In these expressions, ${{\sigma}_{z'}^{2}}$ and ${S_{u'}}$ denote the steady-state variance of $z'$ and the steady-state power spectral density of $u'$ in Fig.~\ref{fig2}, respectively.  	
\end{Theorem}
\begin{proof}
Since $D>{{D}_{inf}}(h)$, there is at least one pair, say $\hat{{E}}$ and $\hat{{D}}$,  satisfying Assumption~\ref{ass2}, rendering the NCS of Fig.~\ref{fig1} SAWSS and producing $\hat{z}$, $\hat{y}$ and $\hat{u}$ where ${{{\sigma}_{\hat{z}}^{2}}\leq{D}}$ holds and
\begin{IEEEeqnarray}{rcl}\label{eq14}
\begin{split}
	\mathcal{R}
	&\geq{{I_{\infty}^{(h)}}(\hat{y}\to\hat{u})}
	\geq{I_{\infty}^{(h)}}({\hat{y}_G}\to\hat{u}_G)\\
	&=\frac{1}{4\pi}\int_{-\pi}^{\pi} \log\big(\frac{{S_{\breve{u}}}(e^{j\omega})}{{\sigma}_{{\hat{\psi}}_G}^{2}}\big)d\omega
\end{split}	
\end{IEEEeqnarray}
can be concluded based on Theorem~\ref{th1}, if the conditions in Lemma~\ref{lemma41} and Lemma~\ref{lemma42} are met. It should be noted that ${S_{\breve{u}}}$ denotes the steady-state power spectral density of $\hat{u}_G$. A coding scheme comprised of linear filters with a unit-gain noisy channel and delay $h$, as follows, can generate $\hat{y}_G$ and $\hat{u}_G$ which satisfy those conditions and keep ${{\sigma}_{\hat{z}_G}^{2}}$ within $({D_{inf}}(h),\infty)$:
   	 \begin{equation}\label{eq15}
   	 	{\hat{u}_{G}}(k)={L_k}({\hat{y}}_{G}^{k-h},{\hat{u}}_{G}^{k-1})+{\hat{\psi}_{G}}(k-h),\qquad k\in{{\mathbb{N}}_0}
   	 \end{equation} 
  where ${\hat{\psi}_{G}}(k)$ represents a Gaussian noise with zero mean and independent of $({\hat{y}}_{G}^{k},{\hat{u}}_{G}^{k-1})$.  Regarding the causality and linearity of ${L_k}$, ${\hat{u}_{G}^{k}}$ can be written as follows:
  \begin{equation}\label{eq16}
  {\hat{u}_{G}^{k}}={Q_k}{\hat{\psi}_{G}^{k-h}}+{P_k}{\hat{y}}_{G}^{k-h},\qquad k\in{{\mathbb{N}}_0}.
  \end{equation} 
 According to the causality in (\ref{eq16}), joint SAWSS-ness of $({\hat{y}_G},\hat{u}_G)$ and transitivity of asymptotic equivalence for products and sum of the matrices noted in \cite{gray2006toeplitz}, the sequences $\{Q_k\}$ and $\{P_k\}$ are asymptotically equivalent to sequences of lower triangular Toeplitz matrices. Moreover, $L_k$ renders the NCS internally stable and well-posed. With all of this in mind, by setting $J=1$ and $B$ as a concatenation of  linear filters with the steady-state behaviour of $L_k$ in (\ref{eq15}) and considering a variance for $\eta$ equal to ${{\sigma}_{{\hat{\psi}}_G}^{2}}$, the system of Fig.~\ref{fig2} will be rendered well-posed and internally stable where  ${S}_{u'}={S}_{\breve{u}}$ and ${{\sigma}_{{\hat{z}}_{G}}^{2}}={{\sigma}_{{z'}}^{2}}$ are resulted.
   \begin{figure}[thpb]
   	\centering
   	\includegraphics[width=6cm,height=3.5cm]{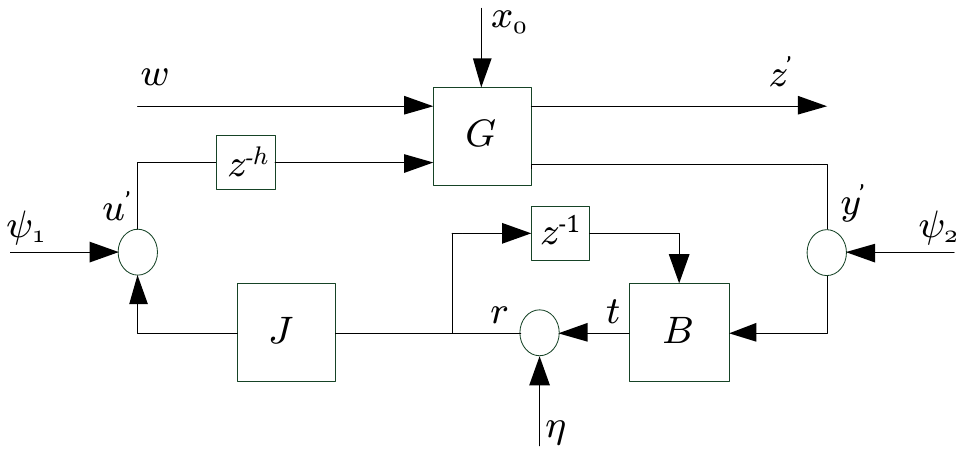}
   	\caption{Equivalent archtiectural viewpoint of internal stability}
   	\label{fig3}
   \end{figure}
 Then according to Lemma~\ref{lemma42}, the following can be concluded:        
 \begin{IEEEeqnarray}{rcl}\label{eq17}
 	\begin{split}
 {{I_{\infty}^{(h)}}(y'\to{u'})}&=\frac{1}{4\pi}\int_{-\pi}^{\pi} \log\big(\frac{{S_{u'}}(e^{j\omega})}{{\sigma}_{\eta}^{2}}\big)d\omega\\
 &=\frac{1}{4\pi}\int_{-\pi}^{\pi} \log\big(\frac{{S_{\breve{u}}}(e^{j\omega})}{{\sigma}_{{\hat{\psi}}_G}^{2}}\big)d\omega,
\end{split}	
\end{IEEEeqnarray}   
which completes the proof.
\end{proof}
\begin{Lemma}\label{lemma43}
Consider the LTI loop of Fig.~\ref{fig2} with fixed ${{\sigma}_{{\eta}}^{2}}\in{{\mathbb{R}}^{+}}$. Define ${\vartheta}_{r}^{'}$ as follows:
\begin{equation}\label{eq18}
{\vartheta}_{r}^{'}(B,J,{{\sigma}_{{\eta}}^{2}})\triangleq\frac{1}{4\pi}\int_{-\pi}^{\pi} \log\big(\frac{{S_{r}}(e^{j\omega})}{{\sigma}_{\eta}^{2}}\big),
\end{equation}
in which ${S_{r}}$ represents the steady-state power spectral density of $r$. Tehn for any $\rho>0$, upon the existence of the pair $(B,J)=({B_{1}},{J_{1}})$ making the system of Fig.~\ref{fig2} internally stable and well-posed, there exist another pair, comprised of the biproper filter ${J_{2}}$ and ${B_{2}}$, which renders the feedback loop of Fig.~\ref{fig2} internally stable and well-posed, preserves the the steady-state power spectral density of $z'$ and satisfies the following:
 \begin{IEEEeqnarray}{rl}\label{eq19}
 	\begin{split}
&{\vartheta}_{r}^{'}({B_{1}},{J_{1}},{{\sigma}_{{\eta}}^{2}})={\vartheta}_{r}^{'}({B_{2}},{J_{2}},{{\sigma}_{{\eta}}^{2}})\\
&={\frac{1}{2}}\log(1+{\frac{{\sigma}_{t}^{2}}{{\sigma}_{\eta}^{2}}}){{\mid}_{(B,J)=({B_2},{J_2})}´}-\rho
 	\end{split}	
 \end{IEEEeqnarray}   	
\end{Lemma}
\begin{proof}
It is well-known that the system of Fig.~\ref{fig2} is well-posed and internally stable if and only if the transfer function $T$ from ${[\eta,w,{\psi}_1,{\psi}_2]^T}$ to ${[z',y',r,u']^T}$ in Fig.~\ref{fig3} belongs to ${\mathcal{R}\mathcal{H}}_{\infty}$.
$T$ is calculated as (\ref{eq27}) (at the top of the next page) where
\begin{equation}\label{eq40}
M={(1-{B_{r}}{z^{-1}}-{G_{22}}J{z^{-h}}{B_y})}^{-1}.
\end{equation}
 By $T_{i}$ and $r_i$, we refer to the transfer-function matrix $T$ and signal $r$ when $B$ and $J$ are set such that $(B,J)=({B_{i}},{J_{i}}),i\in\{1,2\}$. Moreover,  ${B}_{yi}$ and ${B}_{ri}$ represent elements of $B$ ($B=[B_r\quad B_y]$) in the situation where $B={B}_{i}$. Now, consider the following set of filters:
 \begin{IEEEeqnarray}{rl}\label{eq20}
 	\begin{split}
 		{J}_{2}={z^{d_{1}}}{J_{1}}{V^{-1}},\quad {B_{y2}}={z^{{-d_{1}}}}{B_{y1}}\\
 		{B_{r2}}=z(1-(1-{B_{r1}}{z^{-1}}){V^{-1}}),
 	\end{split}	
 \end{IEEEeqnarray}
in which $d_1$ indicates the relative degree of $J_{1}$ and $V\in{{\mathcal{U}}_\infty}$ is chosen in such a way that $V(\infty)=1$. 
Consequently, $J_{2}$ is biproper and $T_{2}$ can be written as follows:
 \begin{IEEEeqnarray}{rl}\label{eq21}
 	\begin{split}
	T_{2}=\mathrm{diag}\{{z^{d_1}}I,{z^{d_1}}I,&V,{z^{d_1}}I\}\times{T_{1}}\times\\
	&{\mathrm{diag}\{I,{z^{-d_1}}I,{z^{-d_1}}I,{z^{-d_1}}I\}}.
 	\end{split}	
 \end{IEEEeqnarray}
So regarding the definition of $d_1$ and properties of $V$, ${T_{2}}\in{\mathcal{R}\mathcal{H}}_{\infty}$ if and only if  ${T_{1}}\in{\mathcal{R}\mathcal{H}}_{\infty}$. Moreover, based on the same argument, using $({B_{2}},{J_{2}})$ would give the same power spectral density for $z'$ as for the case where $({B_{1}},{J_{1}})$ is utilized. We need ${{\Gamma}_{r_{1}}}$ in  $S_{r_{1}}=|{{{\Gamma}_{r_{1}}}}|^{2}$ to be stable, biproper and with all zeros in $\{z\in{\mathbb{C}}:|z|\leq{1}\}$. Let ${\delta_1},\dots,{\delta_{m}}$ represent the zeros of ${{\Gamma}_{r_1}}$  lying on the unit circle. Now for $\zeta\in{(0,1)}$ we define the following: 
  \begin{IEEEeqnarray}{rl}\label{eq22}
  	\begin{split}
 {{{\hat{\Gamma}}}_{r_{1}}}&\triangleq{{{\Gamma}_{r_1}}}{\prod_{i=1}^{m}}z{{(z-{\delta_i})}^{-1}}\\
 {V_\zeta}&\triangleq{({{{\hat{\Gamma}}}_{r_{1}}})}^{-1}{{{{\hat{\Gamma}}}_{r_{1}}}(\infty) }{\prod_{i=1}^{m}}z{{(z-\zeta{\delta_i})}^{-1}}.
 	\end{split}	
 \end{IEEEeqnarray}
   \begin{figure*}[t]
   	\begin{equation}\label{eq27}
   	T=
   	\begin{bmatrix}
   	G_{12}J{z^{-h}}M&G_{11}+G_{12}J{z^{-h}}{B_y}MG_{21}&G_{12}{z^{-h}}(1-{B_{r}}{z^{-1}})M&{G_{12}}J{z^{-h}}{B_y}M\\
   	G_{22}J{z^{-h}}M&G_{21}(1-{B_{r}}{z^{-1}})M&G_{22}{z^{-h}}(1-{B_{r}}{z^{-1}})M&{G_{22}}J{z^{-h}}{B_y}M\\
   	M&{G_{21}}{B_y}M&{G_{22}}{z^{-h}}{B_y}M&{B_y}M\\
   	JM&{G_{21}}J{B_y}M&(1-{B_{r}}{z^{-1}})M&J{B_y}M
   	\end{bmatrix}
   	\end{equation}	
   \end{figure*} 
       Hence, ${{V}_{\zeta}}\in{{\mathcal{U}}_\infty}$ and ${V}_{\zeta}(\infty)=1$ can be deduced for every $\zeta\in{(0,1)}$. By following the same procedure as for the proof of Theorem 5.2 in \cite{silva2011framework}, the existence of $\zeta\in{(0,1)}$ will be shown in such a way that for any ${\rho}>0$, setting $V={V}_{\zeta}$ will give a pair $({B_{2}},{J_{2}})$ that
       satisfies (\ref{eq19}).	
\end{proof}
\newtheorem{Corollary}{Corollary}[section]
\begin{Corollary}\label{cor1}
	If Assumptions~\ref{ass1} and \ref{ass2} hold for the NCS of Fig.~\ref{fig1} and ${D}\in({D_{inf}}(h),\infty)$, then
	\begin{equation}\label{eq23}
	{{\mathcal{R}}(D)}\geq{\frac{1}{2}}\log(1+{{\varphi'}(D)}), {{\varphi'}(D)}\triangleq{{\inf_{{{\sigma}_{z'}^{2}}\leq{D}}}{\frac{{\sigma}_{t}^{2}}{{\sigma}_{\eta}^{2}}}},  
	\end{equation}
where the optimization is done over all LTI filter pairs $(B,J)$ and the noise variance ${{\sigma}_{{\eta}}^{2}}\in{{\mathbb{R}}^{+}}$ making the system in Fig.~\ref{fig2} internally stable and well-posed. Moreover, ${{\sigma}_{t}^{2}}$ and ${{\sigma}_{z'}^{2}}$ represent the steady-state variances of $t$ and $z'$ , respectively.
\end{Corollary}
\begin{proof}
 According to the feasibility of finding ${{{\vartheta}'}_u}(D)$ (shown in Appendix \ref{suba1}), there exist the triplet $({{B}_{\zeta}},1,{{\sigma}_{{\eta}_{\zeta}}^{2}})$, with ${{B}_{\zeta}}$ a proper LTI filter and ${{\sigma}_{{\eta}_{\zeta}}^{2}}\in{\mathbb{R}^{+}}$, that guarantees ${{\sigma}_{z'}^{2}}\leq{D}$ for the system of Fig.~\ref{fig2}. Furthermore, based upon the definition of ${{{\vartheta}'}_u}$ and ${{{\vartheta}'}_r}$  in (\ref{eq13}) and (\ref{eq18}), the following can be derived for any $\zeta>0$:
 \begin{equation}\label{eq24}
 {{{\vartheta}'}_u}(D)+{\zeta}\geq{{\vartheta}_{r}^{'}({{B}_{\zeta}},1,{{\sigma}_{{\eta}_{\zeta}}^{2}})}.
 \end{equation}
So regarding Lemma~\ref{lemma43}, since there exist a biproper filter ${{\tilde{J}}_{\zeta}}$ and a proper one ${{\tilde{B}}_{\zeta}}$ making the LTI feedback loop of Fig.~\ref{fig2} internally stable and well-posed and keeping ${\sigma}_{z'}^{2}$ intact, the following can be concluded:
 \begin{equation}\label{eq25}
 {{{\vartheta}'}_u}(D)+{\zeta}\geq{{{\frac{1}{2}}\log(1+{\frac{{\sigma}_{t}^{2}}{{\sigma}_{\eta}^{2}}}){\mid}_{(B,J,{\sigma}_{\eta}^{2})=({\tilde{B}}_{\zeta},{\tilde{J}}_{\zeta},{{\sigma}_{{\eta}_{\zeta}}^{2}})}}-{\rho}} 
 \end{equation}
Now the proof is completed by noting that (\ref{eq25}) holds for any $\zeta,\rho>0$
\end{proof}
           \begin{figure}[thpb]
           	\centering
           	\includegraphics[width=6cm,height=3.5cm]{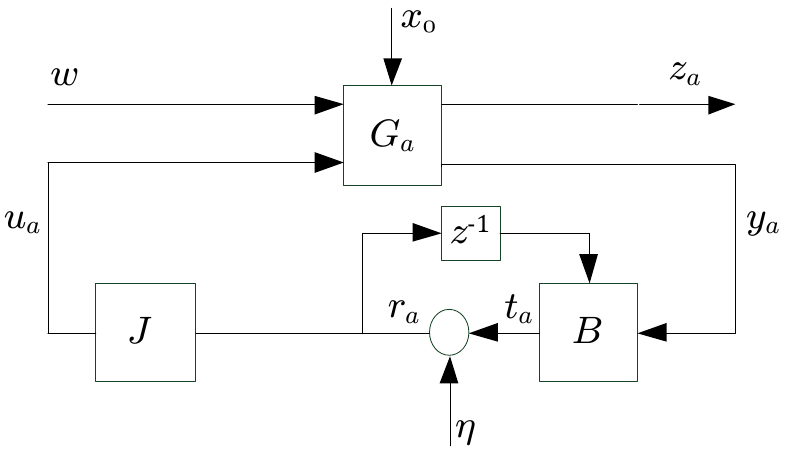}
           	\caption{Auxiliary system for ${{\varphi'}(D)}$ minimization}
           	\label{fig4}
           \end{figure} 
Corollary~\ref{cor1} shows that for the general structure of Fig.~\ref{fig1}, a lower bound on infimum data rate required to achieve ${{{\sigma}_{z}^{2}}\leq{D}}$ is obtained by minimizing the SNR, ${{{\sigma}_{t}^{2}}/{{\sigma}_{\eta}^{2}}}$, over the auxiliary LTI feedback loop of Fig.~\ref{fig2} subject to ${{{\sigma}_{z'}^{2}}\leq{D}}$. To characterize ${{\varphi'}(D)}$, we will mostly use properties of linear systems and some results  on $H_2$ optimization with input-delay.

Consider the auxilary structure of Fig.~\ref{fig4}, where except for shifting the delay block to the plant model, which leads to
\begin{equation}\label{eq26}
{G_a}=\left[
\begin{array}{lr}
G_{11}&{z^{-h}}{G_{12}}\\
{G_{21}}&{z^{-h}}{G_{22}}
\end{array}\right], 
\end{equation}  
the same assumptions as Fig.~\ref{fig2} hold. The NCS of Fig.~\ref{fig4} is internally stable and well-posed if and only if the transfer function ${T_a}$ from ${[\eta,w,{\psi}_1,{\psi}_2]^T}$ to ${[{z_a},{y_a},{r_a},{u_a}]^T}$ in Fig.~\ref{fig5} is a member of ${\mathcal{R}\mathcal{H}}_{\infty}$. Regarding (\ref{eq27}), it can be easily shown that ${T_a}=T$.  So the feedback loop of Fig.~\ref{fig4} and Fig.~\ref{fig2} are equivalent in the sense of internal stability and well-posed-ness. Moreover, the SNR and output variance of the NCS depicted in Fig.~\ref{fig2} can be stated in terms of $H_2$-norms as follows:

 \begin{IEEEeqnarray}{rl}\label{eq28}
 	\begin{split}
{\frac{{\sigma}_{t}^{2}}{{\sigma}_{\eta}^{2}}}&=\norm{{M-1}}_{2}^{2}+{\norm{{B_y}M{G_{21}}}_{2}^{2}}{{\sigma}_{\eta}^{-2}},\\
{{\sigma}_{z'}^{2}}&=\norm{{{G_{11}}+{G_{12}}N{(1-{G_{22}}N)}^{-1}G_{21}}}_{2}^{2}+{\norm{{G_{12}}JM}_{2}^{2}}{{\sigma}_{\eta}^{2}},
\end{split}
 \end{IEEEeqnarray}
where $N\triangleq{J{B_y}{z^{-h}}{(1-{B_{r}}{z^{-1}})}^{-1}}$. Likewise, the following holds for the structure of Fig.~\ref{fig4}:
 \begin{IEEEeqnarray}{rl}\label{eq29}
 	\begin{split}
 {\frac{{\sigma}_{t_a}^{2}}{{\sigma}_{\eta}^{2}}}&=\norm{{{M_a}-1}}_{2}^{2}+{\norm{{B_y}{M_a}{G_{21}}}_{2}^{2}}{{\sigma}_{\eta}^{-2}},\\
 {{\sigma}_{z_a}^{2}}&=\norm{{{G_{11}}+{G_{12}}{z^{-h}}{N_a}{(1-{G_{22}}{z^{-h}}{N_a})}^{-1}G_{21}}}_{2}^{2}+\\
 &\qquad{\norm{{G_{12}}J{M_a}}_{2}^{2}}{{\sigma}_{\eta}^{2}},
 	\end{split}
 \end{IEEEeqnarray}
in which ${N_a}=J{B_y}{(1-{B_{r}}{z^{-1}})}^{-1}$ and ${M_a}=M$. As seen, comparing (\ref{eq28}) and (\ref{eq29}) signifies the equalities $({{{\sigma}_{t}^{2}}/{{\sigma}_{\eta}^{2}}})=({{{\sigma}_{t_a}^{2}}/{{\sigma}_{\eta}^{2}}})$ and ${{\sigma}_{z'}^{2}}={{\sigma}_{z_a}^{2}}$. So every triplet $(B,J,{{\sigma}_{\eta}^{2}})$ that can infimize the SNR while making the system output satisfy ${{{\sigma}_{z'}^{2}}\leq{D}}$ for the NCS of Fig.~\ref{fig4}, can do the same for the the LTI system of our interest, in Fig.~\ref{fig2}, and vice versa. 
    \begin{figure}[thpb]
    	\centering
    	\includegraphics[width=6cm,height=3.5cm]{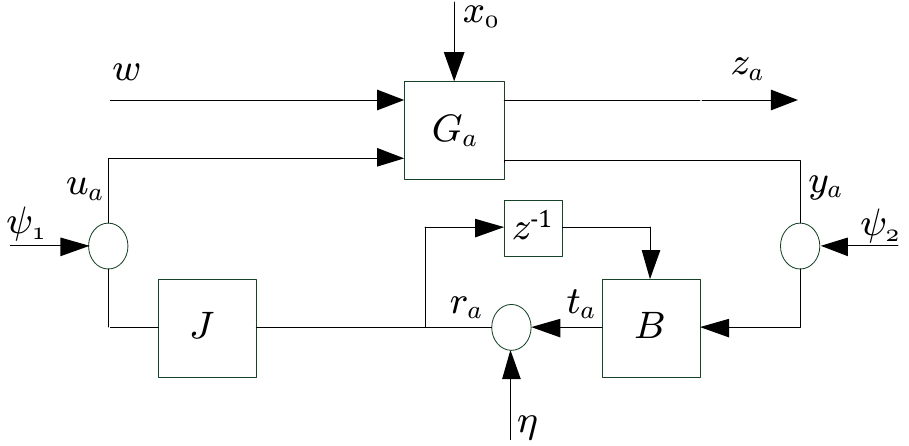}
    	\caption{Stability analysis of the equivalent system}
    	\label{fig5}
    \end{figure}
In other words, the NCSs in Fig.~\ref{fig2} and Fig.~\ref{fig4} are equivalent regarding the SNR-perfomance optimization problem in (\ref{eq23}) as well. This problem is studied for such feedback systems as auxiliary system of Fig.~\ref{fig4} in \cite{silva2013characterization}. So the approach proposed in \cite{silva2013characterization} can be utilized to derive the lower bound on ${{\mathcal{R}}(D)}$ regarding (\ref{eq23}). Consequently, it can be concluded that the problem of finding ${{\varphi'}(D)}$ is equivalent to an SNR-constrained optimal control problem which was proved to be convex. 
 As another result, ${{\varphi'}(D)}$ being a monotonically decreasing fuction of $D$ can be deduced. All in all, the inteplay between the desired performance, the average data rate and  the time delay is characterized through (\ref{eq23}), (\ref{eq28}) and (\ref{eq29}).   
 \section{SIMULATION EXAMPLE}
 Consider the following transfer function representation for the plant $G$ in NCS of Fig.~\ref{fig1}:
 \begin {equation}\label{eq30}
z=	\frac{0.165}{(z-2)(z-0.5789)}(w+u),\quad y=z,
\end{equation}
where $(x_0,w)$ satisfies Assumption \ref{ass1}. Using the results of the previous section, we simulate the lower bound on ${\mathcal{R}}(D)$ obtained in (\ref{eq23}) regarding five different values of delay ($h=\{0,1,2,3,4\}$) and for each $h$, over a range of $D>{D_{inf}}(h)$. Fig.~\ref{fig7} demonstrates the behaviour of the lower bound with respect to $D$ and $h$. Additionally, it shows the operational rates when using scalar uniform quantizer for $h=0$ and $h=4$. First, as expected, ${{\varphi'}(D)}$ in (\ref{eq23}) is a monotonically decreasing function of $D$. Secondly and more importantly, ${D_{inf}}(h)$ increases when $h$ grows. So greater delay yields worse best performance. The most significant outcome is associated with the behaviour of ${\mathcal{R}}(D)$ in (\ref{eq5}) with respect to delay. It can be observed from the curves in Fig.~\ref{fig7} that for a fixed $D$, ${{\varphi'}(D)}$ is increasing in $h$.
Therefore, a delay in the channel forces an increase in the infimum data 
           \begin{figure}[thpb]
           	\centering
           	\includegraphics[width=8cm,height=4cm]{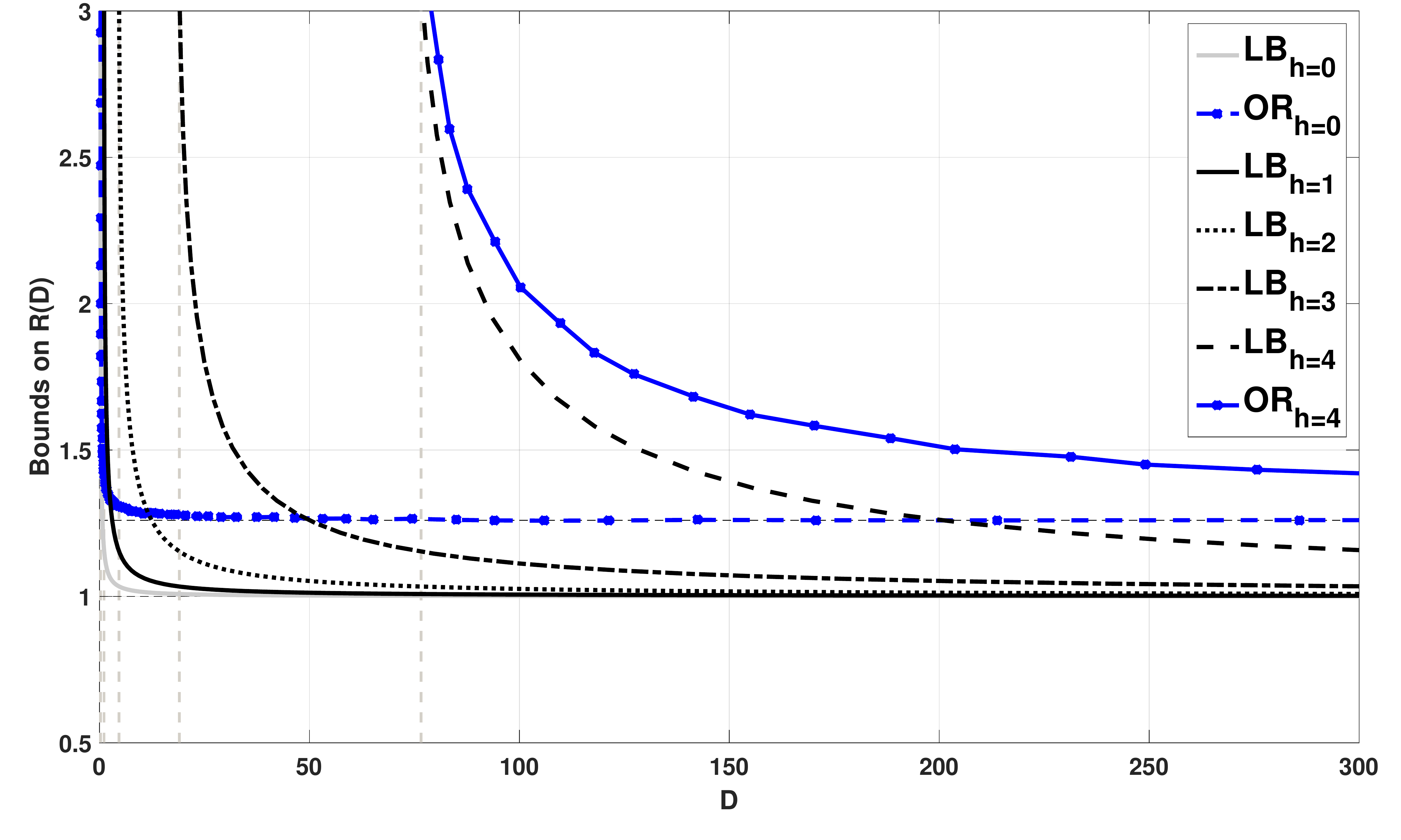}
           	\caption{Bounds on $\mathcal{R}(D)$ in (\ref{eq5}) for different values of time delay $h$}
           	\label{fig7}
           \end{figure} 
           \begin{figure}[thpb]
           	\centering
           	\includegraphics[width=6cm,height=3.5cm]{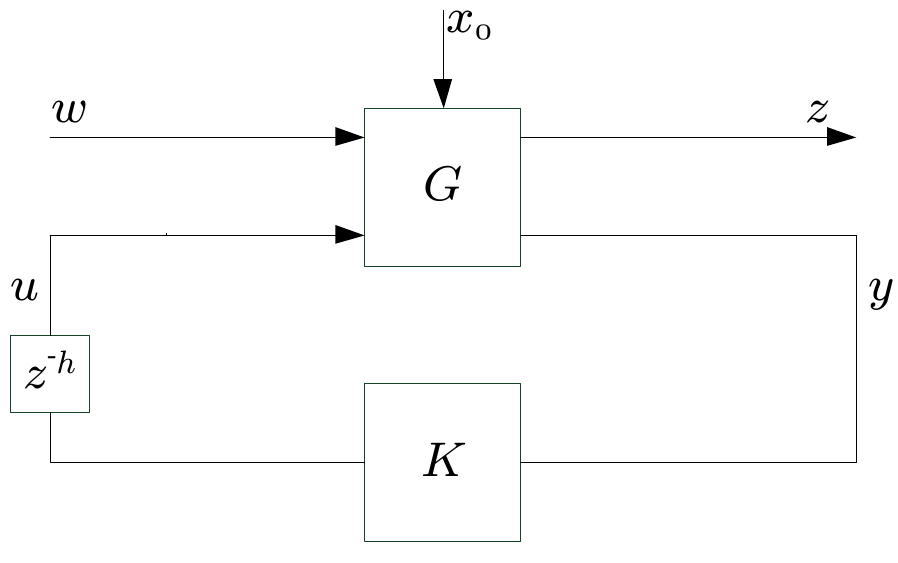}
           	\caption{Standard feedback loop for proving feasibility of finding $\mathcal{R}(D)$}
           	\label{figa1}
           \end{figure} 
rate required to achieve a quadratic level of performance. The greater delay, the higher rate to be spent in order to get a certain level of performance. Another observation is the convergence of the obtained infimal data rates  to the minimum rate required for stabilizability as $D\to\infty$. As seen in Fig.~{\ref{fig7}}, high rates are required to attain the ideal non-networked performance ${D_{inf}(h)}$. Nevertheless, in our case, using only 3 bits can result in a ${\sigma}_{z}^{2}$ fairly close to ${D_{inf}(h)}$. It should be noted that in order to make our contribution more clear, the plant in (\ref{eq30}) has the same dynamics as the one considered in \cite{silva2013characterization}.    	     
 So the curve related to the delay-free case ($h=0$) is identical to the lower bound curve obtained in \cite{silva2013characterization}. Along the lines of \cite{silva2011framework,silva2013characterization}, we now simply replace the AWGN $\eta$ in the independent coding scheme depicted in Fig.~\ref{fig2}, by a uniform scalar quantizer in order to assess the operational performance caused by a simple coding scheme. It is interesting to note that the obtained operational average data rate in Fig.~\ref{fig7} is at most around $0.3$ bits away from the derived lower bound at all performance .       
\section{CONCLUSIONS}

In this paper, rate-constrained networked control systems comprising noisy LTI plants, causal but otherwise arbitrary coding-control schemes and digital noiseless communication channels with time delay, have been studied. For such NCSs, a certain level of performance is attainable if and only if the average data rate does not fall below a minimal value. A lower bound on this infimum rate has been obtained. Through a numerical example, it has been illustrated that the channel's time delay  increases the infimum average data rate needed to achieve a prsecribed level of performance. Moreover, by using a simple scalar quantizer, operational average data rates fairly close (around $0.3$ bits) to the lower bound have been obtained.     
\appendix

\subsection{Feasibility proofs}\label{suba1}
\subsubsection{Feasibility of ${D_{inf}(h)}$}
Suppose that in the standard architecture depicted in Fig.~\ref{figa1}, $G$, $x_0$ and $w$ satisfy Assumption~\ref{ass1} and $K$ follows $u(k)={{\mathcal{K}}_{k}}({y}^{k-h})$.
\begin{figure}[thpb]
	\centering
	\includegraphics[width=6cm,height=3.5cm]{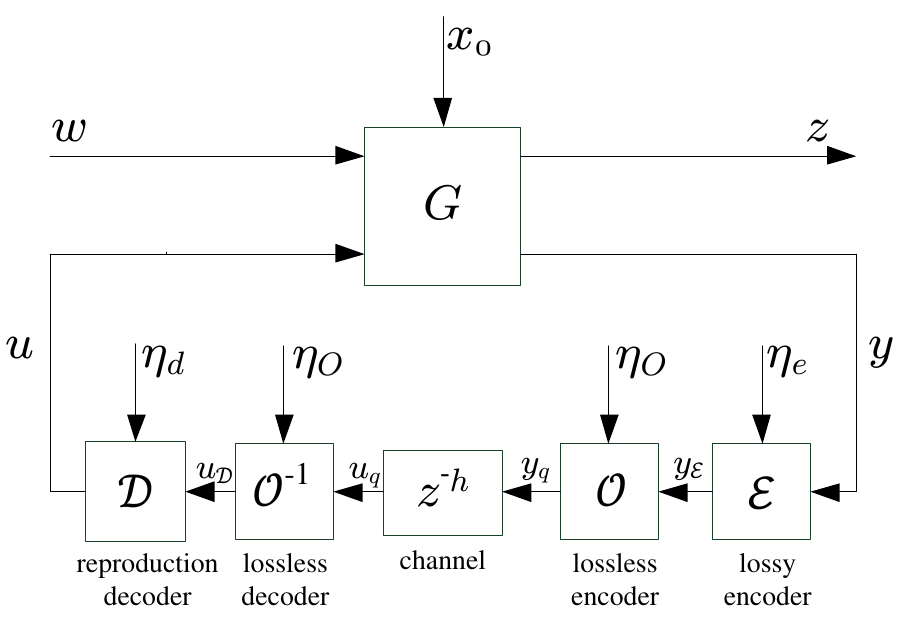}
	\caption{Coding with lossy, lossless and reproduction parts}
	\label{figa2}
\end{figure} 
 Regarding the Gaussianity of $x_0$ and $w$ and the fact that $G$ is LTI, it can be implied from some results in \cite{aastrom2012introduction} that:
\begin{equation*}
	{D}_{inf}(h)=\inf_{K\in{\kappa}}{\sigma}_{z}^{2},
\end{equation*}  
in which ${\sigma}_{z}^{2}$ denotes the variance of output $z$ and ${\kappa}$ is the set of all proper LTI filters which render the system of Fig.~\ref{figa1} internally stable and well-posed. The assumptions considered for $G$ guarantee that finding ${D}_{inf}(h)$ is feasible.
 
\subsubsection{Feasibility of ${{\vartheta}_{u}^{'}(D)}$}
Since  ${D}_{inf}(h)$ can be obtained, for every $\zeta\in(0,{D-{D_{inf}(h)})}$, there exists ${K_1}\in{\kappa}$ which gives ${\sigma}_{z_{1}}^{2}\triangleq{{\sigma}_{z}^{2}}\mid_{K=K_1}\leq{{D_{inf}(h)}+\zeta<D}$ for the system of Fig.~\ref{figa1}. Applying ${K_1}$ to this system results in a stable setting which is a special case of the NCS depicted in Fig.~\ref{fig2} with $J=1$ and $r=t={K_1}y'$ where the steady-state variance of $t$, ${{\sigma}_{t}^{2}}={{\sigma}_{t_{1}}^{2}}$, is finite. Therefore, since ${{K_1}\in{\kappa}}$, it can bring internal stability and well-posed-ness to the feedback loop of Fig.~\ref{fig2} in the presence of any additive noise $\eta$ with steady-sate variance ${{\sigma}_{{\eta}}^{2}}\in{{\mathbb{R}}^{+}}$. 
So ${\sigma}_{z'}^{2}={\sigma}_{z_{1}}^{2}+{{\chi}_{z}}{{\sigma}_{\eta}^{2}}$ and ${\sigma}_{t}^{2}={{\sigma}_{t_{1}}^{2}}+{{\chi}_{t}}{{\sigma}_{\eta}^{2}}$ can be concluded, when taking $\eta$ into account as an AWGN with finite variance ${{\sigma}_{{\eta}}^{2}}$ for the system of Fig.~\ref{fig2}. It should be noted that  ${{\chi}_{t}},{{\chi}_{z}}\geq{0}$ depend only on $K_1$. 
Now by choosing $\zeta=({D-{D_{inf}}(h)})/3$ and the variance ${{\sigma}_{{\eta}}^{2}}=({D-{D_{inf}}(h)})/(3{{\chi}_{z}})$ for the AWGN, there exists ${K_1}\in{\kappa}$ rendering the NCS of Fig.~\ref{fig2} internally stable and well-posed in a way that  ${{\sigma}_{z'}^{2}{\mid}_{(B,J,{\sigma}_{\eta}^{2})=(K_1,1,{\sigma}_{\eta}^{2})}}\leq{{D_{inf}(h)}+{\frac{2}{3}}(D-{D_{inf}(h)})}<D$. Then the following can be obtained for the structure of Fig.~\ref{fig2}:
\begin{equation}\label{eq31}
{{{\frac{{\sigma}_{t}^{2}}{{\sigma}_{\eta}^{2}}}}{\mid}_{(B,J,{\sigma}_{\eta}^{2})=(K_1,1,{\sigma}_{\eta}^{2})}}={\frac{{3{{\sigma}_{t_1}^{2}}{{\chi}_{z}}}}{{D-{D_{inf}}(h)}}}+{{\chi}_{t}}<\infty.
\end{equation}
So regarding Jensen's inequality and concavity of logarithm, it can be deduced that the problem of finding ${{\vartheta}_{u}^{'}(D)}$ in (\ref{eq13}) is feasible for every $D>{D_{inf}}(h)$.
\subsubsection{Feasibility of $\varphi'(D)$}
Immediately from (\ref{eq31}), feasibility of the problem of finding $\varphi'(D)$ in (\ref{eq23}) is inferred for any $D>{D_{inf}}(h)$. 
\subsection{Proof of Theorem~\ref{th1}}\label{suba2}  
The coding scheme in Fig.~\ref{fig1} comprises lossy, lossless and reproduction part as depicted in Fig.~\ref{figa2}. On the encoder side, (\ref{eq3}) can be rewritten as follows:
\begin{IEEEeqnarray}{rl}\label{eq32}
	\begin{split}
{{y_{\mathcal{E}}}(k)}&={\mathcal{E}_k}({y^k},{\eta_{e}^{k}})\\
{y_{q}(k)}&={{\mathcal{O}}_{k}}({y_{\mathcal{E}}^{k}},{\eta_{O}^{k}}),	
    \end{split}
\end{IEEEeqnarray}
in which the side informations belong to well-defined sets, i.e. ${{\eta_{e}}(k)}\in {{{\Theta}_{\mathcal{E}}}(k)}$ and ${{\eta_{O}}(k)}\in {{{\Theta}_{{O}}}(k)}$. Likewise, the decoding scheme in (\ref{eq4}) is given by
\begin{IEEEeqnarray}{rl}\label{eq33}
	\begin{split}
{{u_{\mathcal{D}}}(k)}&=
\begin{cases}
	{\mathcal{O}_{k}^{-1}}({\eta_{O}^{k}}), &0\leq{k}<h,\\
	{\mathcal{O}_{k}^{-1}}({{u_{q}^{k}}},{\eta_{O}^{k}}), &k\geq{h},
\end{cases}\\ 
{u(k)}&={\mathcal{D}_k}({u_{\mathcal{D}}^{k}},{\eta_{d}^{k}}),
	\end{split}	
\end{IEEEeqnarray}
where ${{u_{\mathcal{D}}}(k)}={{y_{\mathcal{E}}}(k-h)}$ holds for $k\geq{h}$. In addition, ${{\eta_{d}}(k)}\in {{{\Theta}_{\mathcal{D}}}(k)}$ and ${{{\Theta}_{\mathcal{D}}}(k)}$ is a well-defined set. The set ${{{\Theta}_{{O}}}(k)}$ is defined as ${{{\Theta}_{{O}}}(k)}\triangleq{{{{\Theta}_{\mathcal{E}}}(k)}\cap{{{\Theta}_{\mathcal{D}}}(k)}}$. Regarding (\ref{eq32})-(\ref{eq33}), the decoder has access to ${y_{\mathcal{E}}^{k-1}}$ and ${\eta_{O}^{k}}$ when receiving ${y_{\mathcal{E}}(k)}$ at time $k+h$. Accordingly, the rate at each time instant is lower-bounded as follows:
\begin{equation}\label{eq34}
{R(k)}\geq{H({{y_{\mathcal{E}}}(k)}\mid{{y_{\mathcal{E}}^{k-1}},{\eta_{O}^{k}}})}.
\end{equation}     
Furthermore, based on the closed-loop dynamics of the NCS depicted in Fig.~\ref{fig1}, its measurement output, $y(k)$, satisfies
 \begin{equation}\label{eq35}
 	{y(k)}={\mathcal{G}_k}({{u}^{k-1}},{w^k},{x(0)}),
 \end{equation}
  where ${\mathcal{G}_k}$ is a linear deterministic mapping. Accordingly, we can derive the following chain of inequalities and identities:
	\begin{IEEEeqnarray}{rl}\label{eq36}
		\begin{split}
   	{R(i)}&\myeqp{H({{y_{\mathcal{E}}}(i)}\mid{{y_{\mathcal{E}}^{i-1}},{\eta_{O}^{i}}})} \\
   	&\myeqq{H({{y_{\mathcal{E}}}(i)}\mid{{y_{\mathcal{E}}^{i-1}},{\eta_{d}^{i+h}}})}\\
   	&\myeqr{H({{y_{\mathcal{E}}}(i)}\mid{{y_{\mathcal{E}}^{i-1}},{\eta_{d}^{i+h}}})-H({{y_{\mathcal{E}}}(i)}\mid{{y_{\mathcal{E}}^{i-1}},{\eta_{d}^{i+h}},y^i})}\\
   	&\myeqs{I({{y_{\mathcal{E}}}(i);y^i}\mid{{y_{\mathcal{E}}^{i-1}},{\eta_{d}^{i+h}}}})\\
   	&\myeqt{I({{y_{\mathcal{E}}}(i);y^i}\mid{{u_{\mathcal{D}}^{i+h-1}},{\eta_{d}^{i+h}}}})\\
   	&\myequ{I({{y_{\mathcal{E}}}(i);y^i}\mid{{u^{i+h-1}},{\eta_{d}^{i+h}}}})\\
   	&\myeqv{I({{u}(i+h);y^i}\mid{{u^{i+h-1}},{\eta_{d}^{i+h}}}})\\
   	&\myeqw{I({{u}^{i+h},{\eta_{d}^{i+h}};y^i})-I({{u}^{i+h-1},{\eta_{d}^{i+h}};y^i})}\\
   	&\myeqx{I({{u}^{i+h};y^i})-I({{u}^{i+h-1},{\eta_{d}^{i+h}};y^i})}\\
   	&\myeqy{I({{u}(i+h);y^i}\mid{{u^{i+h-1}}})-I({{\eta_{d}^{i+h}};y^i}\mid{{u^{i+h-1}}}})\\
   	&\myeqz{I({{u}(i+h);y^i}\mid{{u^{i+h-1}}})}\\
   	\end{split}
   	\end{IEEEeqnarray}
 where (p) is immediately deduced from (\ref{eq34}). Founded upon one property of the entropy and ${{{\Theta}_{{I}}}(k)}\triangleq{{{{\Theta}_{\mathcal{E}}}(k)}\cap{{{\Theta}_{\mathcal{D}}}(k)}}$, (q) is yielded. (r) and (s) follow from positiveness of the entropy and definition of mutual information, respectively. Additionally, (t) is given based on one property of mutual information and ${{u_{\mathcal{D}}}(k)}={{y_{\mathcal{E}}}(k-h)}$; $k\geq{h}$. According to invertibility of the decoder, (u) is concluded. As a result of Assumption \ref{ass1}, (\ref{eq33}) and (\ref{eq35}), Lemma 4.2 in \cite{silva2011framework} holds for the system of our interest in Fig.~\ref{fig1}. From the second claim of the this Lemma,  the Markov chain${{y}^{i+h}-{{u}_{\mathcal{D}}(i+h)}-u(i+h)}$, conditioned upon $({{\eta}_{d}^{i+h}},{{u}_{\mathcal{D}}^{i+h-1}})$, is inferred. Since  ${{u_{\mathcal{D}}}(i)}={{y_{\mathcal{E}}}(i-h)}$; $i\geq{h}$, The Markov chain ${y}^{i+h}-{{y}_{\mathcal{E}}(i)}-u(i+h)$, conditioned upon $({{\eta}_{d}^{i+h}},{{u}_{\mathcal{D}}^{i+h-1}})$, holds. So the Markov chain ${y}^{i}-{{y}_{\mathcal{E}}(i)}-u(i+h)$, conditioned upon $({{\eta}_{d}^{i+h}},{{u}_{\mathcal{D}}^{i+h-1}})$, holds as well. With all this in mind, the validity of (v) is verified regarding the invertibility of the decoder and (M3) in \cite{silva2011framework}. (w), (x) and (y) are caused by the  chain rule of mutual information. (z) is deduced based upon the Markov chain ${{\eta}_{d}^{i+h}}-{u}^{i+h-1}-{y}^{i}$ which is caused by another Markov chain ${{\eta}_{d}^{i+h}}-{u}^{i+h-1}-{y}^{i+h}$ (Lemma 4.2 in \cite{silva2011framework}). Hence, 
\begin{IEEEeqnarray}{rl}\label{eq37}
\begin{split}
 {R(i)}&\geq{I({{u}(i+h);y^i}\mid{{u^{i+h-1}}})}\\
 \Sigma_{i=0}^{k-1}{R(i)}&\geq\Sigma_{i=0}^{k-1}{I({{u}(i+h);y^i}\mid{{u^{i+h-1}}})}\\
 &\myeqaa\Sigma_{i=0}^{k-1-h}{I({{u}(i+h);y^i}\mid{{u^{i+h-1}}})}\\
 &\myeqab\Sigma_{i=-h}^{k-1-h}{I({{u}(i+h);y^i}\mid{{u^{i+h-1}}})}
\end{split}
\end{IEEEeqnarray}
where (aa) holds due to the positive-ness of mutual information and (ab) stems from Assumption~\ref{ass2}. Consequently
\begin{equation} 
{\Sigma}_{i=0}^{k-1}{R(i)}\geq{\Sigma}_{i=0}^{k-1}{I({{u}(i);y^{i-h}}\mid{{u^{i-1}}})},
\end{equation} 
from which (\ref{eq6}) is resulted immediately.
\bibliographystyle{./IEEEbib} 
 \bibliography{./IEEEabrv,./refs}

\begin{thebibliography}{10}

\bibitem{zhang2013network}
Lixian Zhang, Huijun Gao, and Okyay Kaynak,
\newblock ``Network-induced constraints in networked control systems—a
  survey,''
\newblock {\em IEEE Transactions on Industrial Informatics}, vol. 9, no. 1, pp.
  403--416, 2013.

\bibitem{baillieul2007control}
John Baillieul and Panos~J Antsaklis,
\newblock ``Control and communication challenges in networked real-time
  systems,''
\newblock {\em Proceedings of the IEEE}, vol. 95, no. 1, pp. 9--28, 2007.

\bibitem{delchamps1990stabilizing}
David~F Delchamps,
\newblock ``Stabilizing a linear system with quantized state feedback,''
\newblock {\em IEEE Transactions on Automatic Control}, vol. 35, no. 8, pp.
  916--924, 1990.

\bibitem{wong1999systems}
Wing~Shing Wong and Roger~W Brockett,
\newblock ``Systems with finite communication bandwidth constraints--{I}{I}:
  Stabilization with limited information feedback,''
\newblock {\em IEEE Transactions on Automatic Control}, vol. 44, no. 5, pp.
  1049--1053, 1999.

\bibitem{almakhles2015stability}
Dhafer~J Almakhles, Akshya~K Swain, and Nitish~D Patel,
\newblock ``Stability and performance analysis of bit-stream-based feedback
  control systems,''
\newblock {\em IEEE Transactions on Industrial Electronics}, vol. 62, no. 7,
  pp. 4319--4327, 2015.

\bibitem{nair2004stabilizability}
Girish~N Nair and Robin~J Evans,
\newblock ``Stabilizability of stochastic linear systems with finite feedback
  data rates,''
\newblock {\em SIAM Journal on Control and Optimization}, vol. 43, no. 2, pp.
  413--436, 2004.

\bibitem{lupu2015information}
Mircea~F Lupu, Mingui Sun, Fei-Yue Wang, and Zhi-Hong Mao,
\newblock ``Information-transmission rates in manual control of unstable
  systems with time delays,''
\newblock {\em IEEE Transactions on Biomedical Engineering}, vol. 62, no. 1,
  pp. 342--351, 2015.

\bibitem{zhao2015bode}
Yingbo Zhao and Vijay Gupta,
\newblock ``A {B}ode-like integral for discrete linear time-periodic systems,''
\newblock {\em IEEE Transactions on Automatic Control}, vol. 60, no. 9, pp.
  2494--2499, 2015.

\bibitem{ostergaard2016multiple}
Jan {\O}stergaard and Daniel Quevedo,
\newblock ``Multiple descriptions for packetized predictive control,''
\newblock {\em EURASIP Journal on Advances in Signal Processing}, vol. 2016,
  no. 1, pp. 1--16, 2016.

\bibitem{martins2008feedback}
Nuno~C Martins and Munther~A Dahleh,
\newblock ``Feedback control in the presence of noisy channels:
  {“}{B}ode-like{”} fundamental limitations of performance,''
\newblock {\em IEEE Transactions on Automatic Control}, vol. 53, no. 7, pp.
  1604--1615, 2008.

\bibitem{silva2011framework}
Eduardo~I Silva, Milan~S Derpich, and Jan {\O}stergaard,
\newblock ``A framework for control system design subject to average data-rate
  constraints,''
\newblock {\em IEEE Transactions on Automatic Control}, vol. 56, no. 8, pp.
  1886--1899, 2011.

\bibitem{silva2013characterization}
Eduardo~I Silva, Milan Derpich, Jan {\O}stergaard, and Marco Encina,
\newblock ``A characterization of the minimal average data rate that guarantees
  a given closed-loop performance level,''
\newblock {\em IEEE Transactions on Automatic Control}, vol. 61, no. 8, pp.
  2171--2186, 2016.

\bibitem{tanaka2016rate}
Takashi Tanaka, Karl~Henrik Johansson, Tobias Oechtering, Henrik Sandberg, and
  Mikael Skoglund,
\newblock ``Rate of prefix-free codes in {L}{Q}{G} control systems,''
\newblock {\em arXiv preprint arXiv:1604.01227}, 2016.

\bibitem{cover2012elements}
Thomas~M Cover and Joy~A Thomas,
\newblock {\em Elements of information theory},
\newblock John Wiley \& Sons, 2012.

\bibitem{derpich2013fundamental}
Milan~S Derpich, Eduardo~I Silva, and Jan {\O}stergaard,
\newblock ``Fundamental inequalities and identities involving mutual and
  directed informations in closed-loop systems,''
\newblock {\em arXiv preprint arXiv:1301.6427}, 2013.

\bibitem{porat1994digital}
Boaz Porat,
\newblock {\em Digital processing of random signals: Theory and methods},
\newblock Prentice-Hall, Inc., 1994.

\bibitem{gray2006toeplitz}
Robert~M Gray,
\newblock {\em Toeplitz and circulant matrices: A review},
\newblock now publishers inc, 2006.

\bibitem{aastrom2012introduction}
Karl~J {\AA}str{\"o}m,
\newblock {\em Introduction to stochastic control theory},
\newblock Courier Corporation, 2012.

\end{thebibliography}
\addtolength{\textheight}{-12cm}   



\end{document}